\documentclass[11pt]{article}

\usepackage{graphicx}
\usepackage[cmex10]{amsmath}
\usepackage{amsthm}
\usepackage{amssymb}
\usepackage{bm}
\usepackage{latexsym}
\usepackage{bbm}
\usepackage{multirow}
\usepackage{subcaption}
\newtheorem{theorem}{Theorem}
\newtheorem{corollary}{Corollary}
\newtheorem{lemma}{Lemma}
\newtheorem{remark}{Remark}
\newtheorem{example}{Example}
\newtheorem{definition}{Definition}
\newtheorem{assumption}{Assumption}

\newcommand{\E}{\mathrm{E}}
\newcommand{\A}{\mathrm{A}}
\newcommand{\D}{\mathrm{D}}

\newcommand{\dd}{\mathrm{d}}
\newcommand{\one}{\mathbbm{1}}

\newcommand{\st}{\mathrm{st}}

\newcommand{\ds}{\displaystyle}

\allowdisplaybreaks

\title{Queueing Analysis of GPU-Based Inference Servers with Dynamic Batching: 
A Closed-Form Characterization}

\author{Yoshiaki Inoue\footnote{Y.\ Inoue is with
Department of Information and Communications Technology, Graduate School of Engineering, 
Osaka University, Suita 565-0871, Japan. (E-mail: yoshiaki@comm.eng.osaka-u.ac.jp).
}}

\begin{document}
\maketitle

\begin{abstract}

GPU-accelerated computing is a key technology to realize 
high-speed inference servers using deep neural networks (DNNs). 
An important characteristic of GPU-based inference is that the
computational efficiency, in terms of the processing speed and energy
consumption, drastically increases by processing multiple jobs
together in a batch. In this paper, we formulate GPU-based inference
servers as a batch service queueing model with batch-size dependent
processing times. We first show that the energy efficiency of the
server monotonically increases with the arrival rate of inference
jobs, which suggests that it is energy-efficient to operate the
inference server under a utilization level as high as possible within
a latency requirement of inference jobs. We then derive a
closed-form upper bound for the mean latency, which provides a simple
characterization of the latency performance. Through simulation and
numerical experiments, we show that the exact value of the mean latency
is well approximated by this upper bound.
We further compare this upper bound with the latency curve measured in 
real implementation of GPU-based inference servers and we show that
the real performance curve is well explained by the derived simple formula.
\end{abstract}


\section{Introduction}
\label{sec:intro}

Deep neural networks (DNNs) have become increasingly popular 
tools to implement artificial intelligence (AI) related capability,
such as image classification and speech recognition, into mobile applications.
Because executing inference with DNNs is computationally heavy
for mobile devices, inference jobs are typically offloaded to a cloud
(or fog) server. From the server's perspective, many inference jobs
originating a large number of different mobile devices arrive,
and the server should process them within latency requirements of the
applications. To realize high-speed DNN inference, such a server usually 
utilizes the parallel computing capability of a GPU, which largely
accelerate the inference process \cite{nvidia,Xu2018}.

GPU-based inference has an interesting characteristic that 
batching many jobs drastically increases the computing efficiency in
terms of the processing speed and energy consumption 
\cite{nvidia,Crankshaw2017,Xu2018}.
Table \ref{table:resnet} shows measurement results of the computing
performance for two different GPUs (Tesla V100 and Tesla P4) and
precision (FP16/FP32 Mixed precision and INT8), which are reported in 
\cite{nvidia}. A DNN called ResNet-50 is employed for these
measurements, which is a winner of ImageNet Large Scale Visual
Recognition Competition 2015 (ILSVRC2015).
Note that the energy efficiency is represented as the number of
inference jobs per unit time which is able to be processed with
unit power (measured in Watt). We can also interpret this quantity as 
the average number of inference jobs processed with unit energy
(measured in Joule). In each case of Table \ref{table:resnet}, we see
that the throughput and the energy efficiency largely increase by
batching multiple jobs.

\begin{table}[tb]
\caption{Measurement results for inference performance using 
ResNet-50 reported in \cite{nvidia}. See \cite[Pages 23--24]{nvidia}
for more details of the measurement methodology.}
\label{table:resnet}
\begin{subtable}{0.98\textwidth}
\centering
\caption{Tesla V100 (Mixed precision).}
\begin{tabular}{|c|c|c|c|}
\hline
\multirow{2}{*}{Batch Size}
& Throughput
& Average Board
& Throughput/Watt
\\
& [images/sec] &  Power [Watt] & [images/sec/Watt]
\\
\hline
1 & 476 & 120 & 4 
\\
2 & 880 & 109 & 8.1
\\
4 & 1,631 & 132 & 12.4
\\
8 & 2,685 & 153 & 17.5
\\
64 & 5,877 & 274 & 21.4
\\
128 & 6,275 & 285 & 22
\\
\hline
\end{tabular}
\end{subtable}
\mbox{}\vspace{4ex}
\\
\centering
\begin{subtable}{0.98\textwidth}
\centering
\caption{Tesla P4 (INT8).}
\begin{tabular}{|c|c|c|c|}
\hline
\multirow{2}{*}{Batch Size}
& Throughput
& Average Board
& Throughput/Watt
\\
& [images/sec] &  Power [Watt] & [images/sec/Watt]
\\
\hline
1   & 569 & 44 & 12.9
\\
2   & 736 & 44 & 16.7
\\
4   & 974 & 49 & 19.9
\\
8   & 1,291 & 57 & 22.6
\\
64  & 1,677 & 63 & 26.6
\\
128 & 1,676 & 62 & 27
\\
\hline
\end{tabular}
\end{subtable}
\end{table}

Because of this characteristic of GPU-based inference, it is 
efficient for a server to combine multiple inference jobs 
arriving from different devices into a batch, and process them
simultaneously. Such a dynamic batching procedure is indeed supported
by DNN server-application libraries such as TensorFlow-Serving
\cite{tensorflow} and TensorRT Inference Server \cite{tensorrt}.

Despite the importance of GPU inference servers, there have been few
studies focusing on the mathematical analysis of their performance,
which may be due to the fact that the training process has been
given far more weight than the inference process in the machine-learning
community. In \cite{Cai2017}, a predictive framework called
NeuralPower is proposed, which provides a prediction of 
the inference time and energy consumption for each layer based on the
sparse polynomial regression. A prediction method for the inference
time is also proposed in \cite{Dai2019}, which looks up a database of
per-operator execution times and sums them up for all operations
involved. These previous studies focus only on the execution time of a
single inference, which corresponds to the service time in the
queueing theoretic terminology.
To the best of our knowledge, there are no previous studies 
on GPU inference servers that analyze \emph{the system latency}
including the queueing delay.

The main purpose of this paper is to introduce a queueing theoretic
perspective on GPU-based DNN inference systems with dynamic batching.
We formulate an inference server as a batch-service queueing
model with batch-size dependent processing times and we present a
novel analytical method for this model. 
As an initial study, this paper mainly focuses on the derivation of
\emph{a closed-form formula} that can characterize the latency
performance of GPU-based inference servers.
Although the analysis of
batch-service queues is a well-studied subject of the queueing theory
\cite{Bailey54,Briere89,Deb73,Downton55,Downton56,Jaiswal60,Medhi75,Neuts65,Neuts67,Neuts89},
most of them assume that the processing time distribution is
independent of the batch size.
This assumption does not hold for GPU-based inference servers because 
the processing times of inference jobs increase with the batch size.
Batch-service queues with batch-size dependent processing times
are analyzed in \cite{Neuts65}, \cite{Neuts67}, \cite[Section
4.2]{Neuts89}, where computational procedures to numerically obtain
several performance metrics are given. 
In particular, the matrix-analytic method developed in \cite{Neuts89}
provides a unified way to perform an algorithmic analysis of a wide
range of batch service queueing models. 
However, the main weakness of numerical approaches such as matrix
analysis methods is that they are algorithmic in nature and do not
yield closed-form formulas.
We note that batch-size dependent processing times make the system
dynamics complicated and it is difficult to obtain a closed-form
formula even in the M/M/1 model \cite{Curry1985}.

In this paper, we first show that the energy-efficiency of the system
monotonically increases with the arrival rate of inference jobs (i.e.,
the system load), by means of the stochastic comparison techniques
\cite{Mull02, Shak07}. This result suggests that it is
energy-efficient to operate the server under a utilization level as
high as possible within a latency requirement. We then derive a
closed-form upper bound of the mean latency, which provides a simple
characterization to  the latency performance of GPU-based inference
servers.

The key idea of our approach is to model the system as a batch-service
queueing model with \textit{infinite maximum batch size} and batch
processing times that linearly increase with the batch size,
where the latter assumption is a specific feature of GPU-based
inference servers as we will validate in this paper.
Note that the finiteness assumption on the maximum batch size is
essential in approaches based on the matrix-analytic method,
because it is a necessary condition for the system being formulated
by a Markov chain with block upper or lower Hessenberg transition
probability matrix. As we will see, however, the assumptions of the
infinite maximum batch size and linear batch processing times
enable us to derive a simple closed-form upper bound of the mean
latency. Furthermore, it is shown through numerical and simulation experiments
that the mean latency is quite well-approximated by this
closed-form upper bound, even for the case with finite maximum batch size.
Therefore, the derived closed-form upper bound can be regarded as a
simple and accurate approximation formula for the mean latency.
We also conduct experiments using real implementation of 
GPU inference servers based on MLPerf inference benchmark
\cite{mlperf} and we show that the simple formula we derive explains the
real performance curve quite well.

The rest of this paper is organized as follows.
In Section \ref{sec:model}, we introduce the mathematical model
considered in this paper. In Section \ref{sec:analysis},
we first show the monotonicity of the  energy-efficiency with respect
to the system load under a relatively general setting, and then derive
a closed-form upper bound for the mean latency assuming linear batch
processing times. In Section \ref{sec:experiments}, we perform
numerical evaluation to discuss the tightness of the
derived upper bound and to validate the obtained results with real
implementation of GPU servers. Finally, we conclude this paper in
Section \ref{sec:conclusion}.

\section{Model}
\label{sec:model}

We model an inference server with dynamic batching as a single-server
batch-service queueing model with infinite buffer. We assume that
arrivals of inference jobs follow a Poisson process with rate
$\lambda$. The server can process multiple inference
jobs simultaneously in a batch, and processing times of
batches are assumed to be independent following a probability
distribution depending on their batch sizes. 
Let $H^{[b]}(x)$ ($x \geq 0$, $b = 1,2,\ldots$) denote the cumulative
distribution function (CDF) of the processing time for a batch of size
$b$. Let $H^{[b]}$ ($b = 1,2,\ldots$) denote a generic random
variable following the CDF $H^{[b]}(x)$. We define $\mu^{[b]}$ 
($b = 1,2,\ldots$) as the mean throughput (the number of
inference jobs processed per time unit) for a batch size $b$:
\begin{equation}
\mu^{[b]} = \frac{b}{\E[H^{[b]}]}.
\label{eq:mu^b}
\end{equation}
Throughout this paper, we make the following assumption:
\begin{assumption}
\label{assumption:H^b}
\item[(i)] $\mu^{[b_1]} \leq \mu^{[b_2]}$ ($b_1 \leq b_2$), i.e.,
the mean throughput $\mu^{[b]}$ is non-decreasing with the batch size $b$.

\item[(ii)] $\lim_{b \to \infty} \mu^{[b]} > \lambda$.
\end{assumption}
Assumption \ref{assumption:H^b} (i) reflects the characteristic of
the GPU-based inference that the computing efficiency increases with
the batch size. Note that under Assumption \ref{assumption:H^b} (i),
the limit $\lim_{b \to \infty} \mu^{[b]}$ is always well-defined. 
Clearly, Assumption \ref{assumption:H^b} (ii) is a necessary (and
sufficient in the batching scheme described below) condition for the
system to be stable.

In order to construct a tractable model, we assume the following simple
dynamic batching scheme: whenever the server is idle and there is at
least one waiting job in the buffer, all of the waiting jobs
are incorporated into a single batch, and its processing is
immediately initiated. To be more specific, suppose that the server is
idle and the buffer is empty at time $0$. Let $B_n$ ($n = 1,2,\ldots$)
denote the size of the $n$th batch processed after time $0$. 
Also, let $L_{\D,n}$ ($n = 1,2,\ldots$) denote the number of waiting
inference jobs just before the departure of the $n$th batch. For
convenience, we define $L_{\D,0} = 0$. Under the batching scheme
described above, all waiting jobs are put into the next batch, 
so that $B_{n+1} = L_{\D,n}$ if $L_{\D,n} > 0$. 
If $L_{\D,n} = 0$, on the other hand, the $(n+1)$st batch contains only one
inference job which have arrived at the empty system. Therefore, it 
follows that
\begin{equation}
B_{n+1} = L_{\D,n} + \one_{\{L_{\D,n}=0\}},
\quad
n = 0,1,\ldots,
\label{eq:B-by-LD}
\end{equation}
where $\one_{\{\cdot\}}$ denotes an indicator function.

In the next section, we will derive analytical results for the
batch-service queueing system described so far.

\section{Queueing Analysis}
\label{sec:analysis}

\subsection{Preliminaries}

Let $A_n$ ($n = 1,2,\ldots$) denote the number of inference jobs 
arriving in the processing time of the $n$th batch.
By definition, the probability function of $A_n$ ($n = 1,2,\ldots$) is
given by
\begin{align}
\Pr(A_n = k \mid B_n = b) &= 
\int_0^{\infty} \frac{e^{-\lambda x}(\lambda x)^k}{k!}
\dd H^{[b]}(x),
\nonumber
\\
&=
a_k^{[b]},
\quad
k = 0,1,\ldots,
\label{eq:A_n-dist}
\end{align}
where $a_k^{[b]}$ ($k = 0,1,\ldots$, $b = 1,2,\ldots$)
is defined as
\begin{equation}
a_k^{[b]} = \int_0^{\infty} \frac{e^{-\lambda x}(\lambda x)^k}{k!}
\dd H^{[b]}(x).
\label{eq:A-dist}
\end{equation}
It is readily verified that the number of waiting jobs $L_{\D,n}$ 
($n = 1,2,\ldots$) at the $n$th processing completion satisfies
\[
L_{\D,n} = A_n,
\]
so that we obtain from (\ref{eq:B-by-LD}), 
\begin{equation}
B_{n+1} = A_n + \one_{\{A_n=0\}}.
\label{eq:B-by-A}
\end{equation}
It then follows from (\ref{eq:A_n-dist}) and (\ref{eq:B-by-A})
that the sequence of processed batch sizes $(B_n)_{n=1,2,\ldots}$
forms a discrete-time Markov chain on state space $\{1,2,\ldots\}$,
whose transition probability matrix $\bm{P}$ is given by
\begin{equation}
\bm{P} 
=
\begin{pmatrix}
a_0^{[1]} +  a_1^{[1]} & a_2^{[1]} & a_3^{[1]}& \cdots
\\                                           
a_0^{[2]} +  a_1^{[2]} & a_2^{[2]} & a_3^{[2]}& \cdots
\\                                           
a_0^{[3]} +  a_1^{[3]} & a_2^{[3]} & a_3^{[3]}& \cdots
\\
\vdots & \vdots & \vdots & \ddots
\end{pmatrix}.
\label{eq:P=}
\end{equation}
Note that this Markov chain is of GI/G/1-type, i.e, there is no skip-free
structure in the transition matrix $\bm{P}$. In general, it is
difficult to characterize the exact stationary distribution of the
GI/G/1-type Markov chain, and one has to resort to numerical
approximation methods such as the truncation techniques \cite{Gibson87,Tweedie98,Liu10}.
As we will see in Section \ref{ssec:special_case}, however,
we can obtain a closed-form upper bound of the mean latency, by
assuming linearly increasing batch processing times.

In the rest of this subsection, we derive some basic relations among
key performance metrics in steady state. Let $B$ denote a generic random
variable following the stationary distribution of $(B_n)_{n=1,2,\ldots}$.
Let $L$ denote a generic random variable for the stationary
number of inference jobs in the system at an arbitrary time instant.
Further let $A^{[b]}$ ($b=1,2,\ldots$) denote a generic random
variable following the probability function $a_k^{[b]}$ ($k =
0,1,\ldots$). It is readily verified from (\ref{eq:A-dist}) that
the first two moments of $A^{[b]}$ are given by
\begin{equation}
\E[A^{[b]}] = \lambda \E[H^{[b]}],
\quad
\E[(A^{[b]})^2] = \lambda^2 \E[(H^{[b]})^2].
\label{eq:Ab-moments}
\end{equation}

We define $\pi(z)$ and $a^{[b]}(z)$ ($|z| \leq 1$) as
the probability generating functions (PGFs) of $L$ and $A^{[b]}$:
\begin{align}
\pi(z) &= \E[z^L] = \sum_{n=0}^{\infty} \Pr(L=n) z^n,
\label{eq:pi(z)-def}
\\
a^{[b]}(z) &= \E[z^{A^{[b]}}] = \sum_{n=0}^{\infty} \Pr(A^{[b]}=n)
z^n.
\nonumber
\end{align}

\begin{lemma}

$\pi(z)$ ($|z| \leq 1$) satisfies
\begin{equation}
\pi(z) 
=
\sum_{b=1}^{\infty} 
\frac{b\Pr(B=b)}{\E[B]}
\cdot
\frac{1-z^b}{b(1-z)}
\cdot
a^{[b]}(z).
\label{eq:pi(z)}
\end{equation}

\end{lemma}

\begin{proof}

\begin{figure}[tb]
\centering
\includegraphics[scale=0.8]{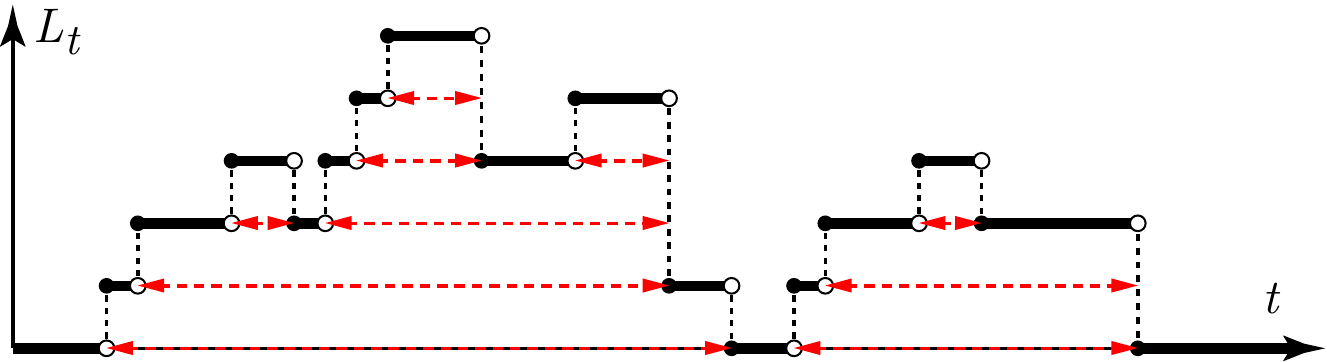}
\caption{A sample-path of the queue-length process $L_t$.}
\label{fig:levelcross}
\end{figure}

Let $L_t$ ($t \geq 0$) denote the number of inference jobs in the
system at time $t$. By definition, each sample-path of $(L_t)_{t \geq 0}$ 
is given by a step function with unit upward jumps (arrivals of
customers) and downward jumps of magnitude $B_n \in \{1,2,\ldots\}$
(completions of batch processing). For convenience, we assume that each
sample-path of $(L_t)_{t \geq 0}$ is constructed so that it is
right-continuous with left limits. Because the system is stable,
there is a one-to-one correspondence between an upward jump and
\textit{the contribution of an inference job to a downward jump} 
(see Fig.\ \ref{fig:levelcross}). To be more specific, 
let $t_n$ and $t_n'$ denote the arrival and departure times of the
$n$th arriving job ($n = 1,2,\ldots$). 
We define $\hat{L}_{\A,n}$ and $\hat{L}_{\D,n}$ as 
\begin{align}
\hat{L}_{\A,n} &= \lim_{\delta t \to 0+}L_{t_n-\delta t},
\nonumber
\\
\hat{L}_{\D,n} &= |\{k > n;\, t_k' = t_n'\}| + L_{t_n'},
\label{eq:hat-L-D-def}
\end{align}
i.e., $\hat{L}_{\A,n}$ denotes the number of inference jobs in the
system seen by the $n$th inference job on arrival, and $\hat{L}_{\D,n}$
denotes the number of inference jobs arrived in the sojourn time
of the $n$th inference job which are in the system just before its
departure. It is then verified that for each sample path,
there is a bijection  $\psi: \{1,2,\ldots\} \to \{1,2,\ldots\}$ such
that $\hat{L}_{\A,n} = \hat{L}_{\D,\psi(n)}$. 

Let $\hat{L}_{\A}$ (resp.\ $\hat{L}_{\D}$) denote a generic random
variable for $\hat{L}_{\A,n}$ (resp.\ $\hat{L}_{\D,n}$) in steady
state. Owing to PASTA and the observation above, we obtain
\begin{equation}
L =_{\st} \hat{L}_{\A} =_{\st} \hat{L}^{\D},
\label{eq:L=L^D}
\end{equation}
where $=_\st$ denotes equality in distribution.
We then consider the distribution of $\hat{L}^{\D}$ to prove 
(\ref{eq:pi(z)}).

Let $\hat{B}$ denote a generic random variable for the size of 
a batch in which a randomly chosen inference job is processed.
It is readily verified that $\hat{B}$ follows the length-biased batch
size distribution, i.e., 
\begin{equation}
\Pr(\hat{B}=b) = \frac{b\Pr(B=b)}{\E[B]}.
\label{eq:hatB-dist}
\end{equation}
We then obtain from (\ref{eq:hat-L-D-def}) and (\ref{eq:L=L^D}),
\[
\pi(z) = \E[z^{\hat{L}_{\D}}]
=
\sum_{b=1}^{\infty} 
\Pr(\hat{B}=b) \left(\sum_{k=0}^{b-1} \frac{1}{b} \cdot z^k\right)
a^{[b]}(z),
\]
which implies (\ref{eq:pi(z)}).
\end{proof}

Let $W$ denote the latency (sojourn time) of a randomly chosen
inference job. We define $H$ (resp.\ $\hat{H}$) as a generic random
variable for the processing time of a randomly chosen batch (resp.\ a
randomly chosen inference job). Note that the distributions of $H$ and $\hat{H}$
are given by (cf.\ (\ref{eq:hatB-dist}))
\begin{align}
\Pr(H \leq x)
&=
\sum_{b=1}^{\infty} \Pr(B=b) \Pr(H^{[b]} \leq x),
\label{eq:H-dist}
\\
\Pr(\hat{H} \leq x) 
&= 
\sum_{b=1}^{\infty} \frac{b\Pr(B=b)}{\E[B]}
\cdot
\Pr(H^{[b]} \leq x).
\label{eq:hatH-dist}
\end{align}

\begin{lemma}
\label{lemma:EW}
The mean latency $\E[W]$ is given by
\begin{equation}
\E[W] 
= 
\frac{\E[B^2]-\E[B]}{2\lambda\E[B]}
+
\E[\hat{H}].
\label{eq:EW}
\end{equation}
\end{lemma}

\begin{proof}
Taking the derivative of (\ref{eq:pi(z)}) and letting $z \to 1$, 
we have
\begin{align}
\E[L] 
&= 
\sum_{b=1}^{\infty} 
\frac{b\Pr(B=b)}{\E[B]}
\left(
\frac{b-1}{2} + \lambda \E[H^{[b]}]
\right)
\nonumber
\\
&=
\frac{\E[B^2]-\E[B]}{2\E[B]}
+
\lambda \E[\hat{H}],
\label{eq:EL}
\end{align}
where we used (\ref{eq:Ab-moments}) in the first equality and (\ref{eq:hatH-dist})
in the second equality. (\ref{eq:EW}) thus follows from Little's law $\E[L]=\lambda \E[W]$.
\end{proof}
\begin{remark}
We can verify that the first term (resp.\ the second term) on the
right-hand side of (\ref{eq:EW}) represents the mean waiting (resp.\
processing) time of a randomly chosen inference job.
\end{remark}

\subsection{Monotonicity of the Energy Efficiency}
\label{ssec:energy}

In this subsection, we show that the larger the system load, the
more energy efficient this system is, under some additional assumptions.
Let $c^{[b]}$ ($b = 1,2,\ldots$) denote the amount of
energy consumed for processing a batch of size $b$.
$c^{[b]}$ is calculated from Table \ref{table:resnet} by the
product of the average board power and the batch processing time
(i.e., the batch size divided by the throughput). For each case in
Table \ref{table:resnet}, $c^{[b]}$ is well-fitted by a linear
function (with the least squares method, we have
the coefficient of determination $R^2 \simeq 0.99978$
for Tesla V100 and $R^2 \simeq 0.99998$ for Tesla P4). 
See Fig.\ \ref{fig:c-resnet} for $c^{[b]}$ plotted as a function of $b$. 

\begin{figure}[tb]
\centering
\includegraphics{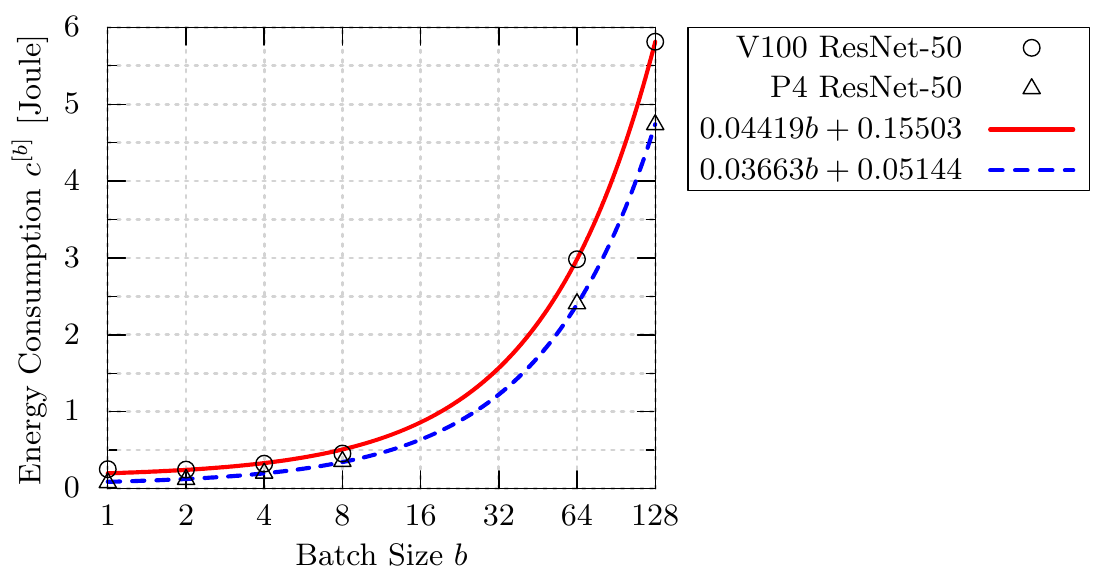}
\caption{The amount of energy consumption for processing a batch 
(calculated from Table \ref{table:resnet}).}
\label{fig:c-resnet}
\end{figure}

We thus make the following assumption on $c^{[b]}$:

\begin{assumption}
\label{assumption:c^b}
$c^{[b]}$ ($b = 1,2,\ldots$) is given by
\begin{equation}
c^{[b]} = \beta b + c_0,
\label{eq:c^b}
\end{equation}
for some $\beta > 0$ and $c_0 \geq 0$.
\end{assumption}

In steady state, the server processes $\lambda/\E[B]$ batches 
per unit time with energy consumption $\sum_{b=1}^{\infty}
\Pr(B=b)c^{[b]}$ on average. We then define the average energy
efficiency $\eta$ of the system as
\begin{equation}
\eta := \frac{\lambda}{\ds \frac{\lambda}{\E[B]} \sum_{b=1}^{\infty}
\Pr(B=b)c^{[b]}},
\label{eq:eta-def}
\end{equation}
i.e., the mean number of inference jobs processed with unit energy.
Under Assumption \ref{assumption:c^b}, (\ref{eq:eta-def}) is rewritten as
\begin{equation}
\eta = \frac{1}{\ds \beta + c_0/\E[B]}.
\label{eq:eta}
\end{equation}

In what follows, we show that the energy efficiency $\eta$ is 
non-decreasing with respect to the arrival rate $\lambda$.
To establish this monotonicity result for $\eta$,
we need an additional assumption on the batch processing 
time distribution $H^{[b]}(x)$ ($b = 1,2,\ldots$): 

\begin{definition}[{\cite[Eq.\ (1.A.1)]{Shak07}}]
Let $X$ and $Y$ denote non-negative random variables.
$X$ is said to be smaller than $Y$ in the usual stochastic order
if and only if
\[
\Pr(X > x) \leq \Pr(Y > x),
\quad
\mbox{for all $x \geq 0$}.
\]
\end{definition}

\begin{remark}[{\cite[Eq.\ (1.A.7)]{Shak07}}]
\label{remark:st-order}
$X \leq_{\st} Y$ holds if and only if
\[
\E[\phi(X)] \leq \E[\phi(Y)],
\]
for any non-decreasing function $\phi(x)$ ($x \geq 0$) provided the
expectations exist. In particular, $X \leq_{\st} Y \Rightarrow
\E[X] \leq \E[Y]$.
\end{remark}
\begin{assumption}
\label{assumption:H^b-st-order}
$H^{[b]} \leq_{\st} H^{[b']}$ holds for any $b \leq b'$.
\end{assumption}

Although Assumption \ref{assumption:H^b-st-order} is a strong
assumption on the batch processing time distribution,
it is reduced to the condition about only their the mean value in
several probability distributions, as shown in the following example:

\begin{example}

In the following cases, we have 
$\E[H^{[b]}] \leq \E[H^{[b']}] \Rightarrow H^{[b]} \leq_{\st}
H^{[b']}$ (cf.\ Remark \ref{remark:st-order}):

\begin{itemize}
\item[(a)] $H^{[b]}$ ($b = 1,2,\ldots$) follows an exponential distribution, i.e., 
\[
\Pr(H^{[b]} > x) = e^{-x/\E[H^{[b]}]},
\quad
x \geq 0.
\]

\item[(b)] $H^{[b]}$ ($b = 1,2,\ldots$) follows a gamma distribution
with a fixed coefficient of variation $c$, i.e.,
\[
\Pr(H^{[b]} > x) = 1 - \frac{\gamma(1/c^2, x/(c^2\E[H^{[b]}] )}{\Gamma(1/c^2)},
\quad
x \geq 0.
\]
where $\Gamma(x)$ and $\gamma(x,y)$ denotes the gamma function and the
lower incomplete gamma function. 

\item[(c)] $H^{[b]}$ ($b = 1,2,\ldots$) takes a constant value,
i.e., 
\[
\Pr(H^{[b]} > x) = \one_{\{\E[H^{[b]}] > x\}},
\quad
x \geq 0.
\]
\end{itemize}
\end{example}

Let $B^{\langle\lambda\rangle}$ and $\eta^{\langle\lambda\rangle}$ 
($\lambda > 0$) denote the stationary batch size and the energy
efficiency represented as functions of the arrival rate $\lambda$.

\begin{theorem}
\label{theorem:B-increasing}
Under Assumption \ref{assumption:H^b-st-order},
the stationary batch size $B^{\langle\lambda\rangle}$ 
($\lambda > 0$) increases with the arrival rate $\lambda$ in the usual
stochastic order, i.e.,
\begin{equation}
B^{\langle\lambda_1\rangle} \leq_{\st} B^{\langle\lambda_2\rangle},
\quad
\mbox{for any $\lambda_1 \leq \lambda_2$}.
\label{eq:B-increasing}
\end{equation}
\end{theorem}
\begin{proof}
Let $P^{{\langle\lambda_m\rangle}}$ ($m=1,2$) denote the
transition probability matrix of $(B_n)_{n=1,2,\ldots}$ 
given $\lambda = \lambda_m$, and let $p_{i,j}^{{\langle\lambda_m\rangle}}$ 
($i,j = 1,2,\ldots$) denote the $(i,j)$th element of $P^{{\langle\lambda_m\rangle}}$.
To prove (\ref{eq:B-increasing}), it is sufficient to show
that in the sense of usual stochastic order,
the probability distribution $p_{i,\cdot}^{{\langle\lambda_1\rangle}}$
increases with $i$, and $p_{i,\cdot}^{{\langle\lambda_1\rangle}}$ is
smaller than $p_{i,\cdot}^{{\langle\lambda_2\rangle}}$ 
\cite[Pages 186--187]{Mull02}, i.e.,
\begin{equation}
\sum_{j=k}^{\infty} p_{i,j}^{\langle\lambda_1\rangle}
\leq
\sum_{j=k}^{\infty} p_{i',j}^{\langle\lambda_1\rangle},
\;\;
i \leq i',\,
k = 1,2,\ldots,
\label{eq:P-monotone}
\end{equation}
and
\begin{equation}
\sum_{j=k}^{\infty} p_{i,j}^{\langle\lambda_1\rangle}
\leq
\sum_{j=k}^{\infty} p_{i,j}^{\langle\lambda_2\rangle},
\;\;
i = 0,1,\ldots,\,
k = 1,2,\ldots.
\label{eq:P-comparison}
\end{equation}
Using (\ref{eq:P=}), we rewrite (\ref{eq:P-monotone}) and
(\ref{eq:P-comparison}) as
\begin{equation}
\sum_{j=k}^{\infty} a_j^{[i],\langle\lambda_1\rangle}
\leq
\sum_{j=k}^{\infty} a_j^{[i'],\langle\lambda_1\rangle},
\;\;
i \leq i',\,
k = 2,3,\ldots,
\label{eq:P-monotone-a}
\end{equation}
and
\begin{equation}
\sum_{j=k}^{\infty} a_j^{[i],\langle\lambda_1\rangle}
\leq
\sum_{j=k}^{\infty} a_j^{[i],\langle\lambda_2\rangle},
\;\;
i = 1,2,\ldots,\,
k = 2,3,\ldots,
\label{eq:P-comparison-a}
\end{equation}
where $a_j^{[i],\langle\lambda_m\rangle}$ ($m=1,2$) is defined as 
(cf.\ (\ref{eq:A-dist}))
\[
a_j^{[i],\langle\lambda_m\rangle}
:=
\int_0^{\infty} \frac{e^{-\lambda_m x}(\lambda_m x)^j}{j!}
\dd H^{[i]}(x).
\]
Let $A^{[i],\langle\lambda_m\rangle}$ ($m=1,2$) denote a generic
random variable satisfying $\Pr(A^{[i],\langle\lambda_m\rangle} = j) 
= a_j^{[i],\langle\lambda_m\rangle}$ ($j = 0,1,\ldots$).

Note that the Poisson distribution is increasing in the usual
stochastic order with respect to its mean \cite[Example 8.A.2]{Shak07},
so that we have from Assumption \ref{assumption:H^b-st-order} and 
\cite[Theorem 1.A.6]{Shak07},
\[
A^{[i],\langle\lambda_1\rangle} 
\leq_{\st}
A^{[i'],\langle\lambda_1\rangle},
\quad
i \leq i',
\]
and from \cite[Theorem 1.A.3 (d)]{Shak07},
\[
A^{[i],\langle\lambda_1\rangle} 
\leq_{\st}
A^{[i],\langle\lambda_2\rangle},
\quad
\lambda_1 \leq \lambda_2,\, i = 0,1,\ldots.
\]
Therefore, we obtain (\ref{eq:P-monotone-a}) and
(\ref{eq:P-comparison-a}), which completes the proof.
\end{proof}

\begin{corollary}
\label{corollary:sigma-increase}
Under Assumptions \ref{assumption:c^b} and
\ref{assumption:H^b-st-order},
the energy efficiency $\eta^{\langle\lambda\rangle}$ is
non-decreasing with the arrival rate $\lambda$.
\end{corollary}
\begin{proof}
Corollary \ref{corollary:sigma-increase} immediately follows from
Theorem \ref{theorem:B-increasing} and (\ref{eq:eta}).
\end{proof}

Corollary \ref{corollary:sigma-increase} suggests that 
it is energy-efficient to operate the inference server under a
utilization level as high as possible within a latency requirement of
inference jobs. In the following subsection, we derive a
closed-form upper bound of the mean latency, assuming 
linearly increasing batch processing times.

\subsection{Deterministic Linear Batch Processing Times}
\label{ssec:special_case}

\begin{figure}[tb]
\centering
\includegraphics{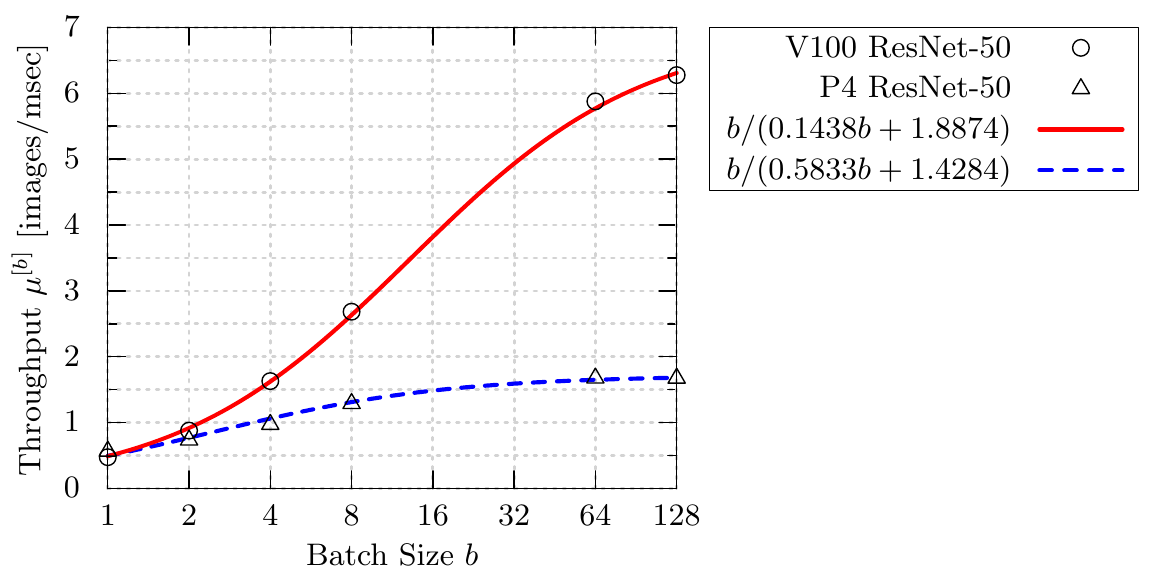}
\caption{Throughput characteristics in Table \ref{table:resnet} 
and corresponding curves plotted by Eq.\ (\ref{eq:mu^b-special}).
}
\label{fig:mu_resnet}
\end{figure}

In Lemma \ref{lemma:EW}, we showed that the mean latency $\E[W]$ is
given in terms of the stationary distribution of the Markov chain
$(B_n)_{n=1,2,\ldots}$ of batch sizes. As mentioned above, 
an exact analysis of the stationary distribution of the GI/GI/1 type
Markov chain $(B_n)_{n=1,2,\ldots}$ is difficult, and only its
numerical approximation is known in the literature.

In this subsection, it is shown that we can obtain a closed-form upper
bound of the mean latency $\E[W]$ by assuming a specific structure in batch
processing times. Specifically, we make the following assumption
throughout this subsection:

\begin{assumption}
\label{assumption:special}
\item The batch processing time $H^{[b]}$ ($b = 1,2,\ldots$) takes a
constant value equal to $\tau^{[b]}$, which is given by 
\begin{equation}
\tau^{[b]} = \alpha b + \tau_0,
\label{eq:tau^b}
\end{equation}
for some $\alpha > 0$ and $\tau_0 \geq 0$.
\end{assumption}

The deterministic distribution is a natural choice 
to model batch inference times because most DNNs 
take a vector of fixed size (input dimension times the batch size) as
its input, and the output is computed by applying a predefined
sequence of operations to it such as matrix multiplications and 
non-linear activation functions, so that the computational steps are
invariant regardless of the input vector. 
Furthermore, we see that the linearity assumption (\ref{eq:tau^b})
is consistent with the measurement results in Table
\ref{table:resnet}: with the least squares method, we have 
the coefficient of determination $R^2 \simeq 0.99975$ 
(resp.\ $R^2 \simeq 0.99986$) with 
$\alpha = 0.1438$ and $\tau_0=1.8874$ 
(resp.\ $\alpha = 0.5833$, $\tau_0 = 1.4284$)
for batch processing times calculated from the data in Table
\ref{table:resnet} (a) (resp.\ Table \ref{table:resnet} (b))
by dividing batch sizes by throughputs (cf.\ (\ref{eq:mu^b})).
Note that under Assumption \ref{assumption:special},
the throughput $\mu^{[b]}$ ($b = 1,2,\ldots$) is written as
\begin{equation}
\mu^{[b]} = \frac{b}{\alpha b + \tau_0}.
\label{eq:mu^b-special}
\end{equation}
As shown in Fig.\ \ref{fig:mu_resnet}, the throughput characteristics
in Table \ref{table:resnet} are well-fitted by this simple rational function.
More validation results on this assumption will be provided 
later in Section \ref{sec:experiments}, using real implementation of
GPU inference servers.

We can readily verify from (\ref{eq:mu^b-special}) that 
Assumption \ref{assumption:special} ensures Assumption
\ref{assumption:H^b} (i). Furthermore, (\ref{eq:mu^b-special}) implies
\[
\lim_{b \to \infty} \mu^{[b]} = \frac{1}{\alpha},
\]
so that the stability condition stated in Assumption
\ref{assumption:H^b} (ii) is rewritten as
\begin{equation}
\rho := \lambda \alpha < 1.
\label{eq:rho}
\end{equation}
In view of this relation, the normalized load $\rho$ 
represents the ratio of the arrival rate to the server's processing
capacity, which corresponds to the traffic intensity in ordinary
single-server queueing models.

Assumption \ref{assumption:special} simplifies the analysis
mainly because under this assumption,
$\E[H]$, $\E[H^2]$, and $\E[\hat{H}]$ 
(see (\ref{eq:H-dist}) and (\ref{eq:hatH-dist}))
are given in terms of the first two moments $\E[B]$ and $\E[B^2]$ of
the stationary batch size distribution:
\begin{align}
\E[H] 
&=
\alpha\E[B]+\tau_0,
\label{eq:EH-special}
\\
\E[H^2]
&=
\alpha^2 \E[B^2] + 2 \alpha \tau_0 \E[B] +  \tau_0^2,
\label{eq:EH2-special}
\\
\E[\hat{H}] 
&= 
\alpha \cdot \frac{\E[B^2]}{\E[B]} + \tau_0.
\label{eq:hatH-special}
\end{align}

\begin{lemma}
Let $A$ denote a generic random variable for $A_n$ in steady state
(cf.\ (\ref{eq:A_n-dist})):
\[
\Pr(A = k) = \sum_{b=1}^{\infty} \Pr(B=b) a_k^{[b]},
\quad
k = 0,1,\ldots.
\]
Under Assumption \ref{assumption:special},
$\E[B]$ and $\E[B^2]$ are given in terms of $\Pr(A=0)$ by
\begin{align}
\E[B] &= \frac{\lambda \tau_0 + \Pr(A=0)}{1-\lambda \alpha},
\label{eq:EB-special}
\\
\E[B^2] 
&=
\frac{(1+2\lambda^2\alpha\tau_0)\E[B]  + \lambda^2 \tau_0^2}
{1-\lambda^2 \alpha^2}.
\label{eq:EB2-special}
\end{align}
\end{lemma}
\begin{proof}

From (\ref{eq:Ab-moments}) and (\ref{eq:H-dist}), we have 
\begin{align}
\E[A] &= \sum_{b=1}^{\infty} \Pr(B=b) \lambda \E[H^{[b]}] 
= \lambda \E[H],
\label{eq:EA-special}
\\
\E[A^2] 
&= 
\sum_{b=1}^{\infty} \Pr(B=b) \left(\lambda\E[H^{[b]}] + \lambda^2 \E[(H^{[b]})^2]\right)
\nonumber
\\
&=
\lambda\E[H] + \lambda^2 \E[H^2].
\label{eq:EA2-special}
\end{align}
It then follows from (\ref{eq:B-by-A}), (\ref{eq:EH-special}), and (\ref{eq:EA-special}) that
\begin{align*}
\E[B] &= \E[A]+\Pr(A=0)
\\
&=
\lambda(\alpha \E[B] + \tau_0)+\Pr(A=0),
\end{align*}
so that we obtain (\ref{eq:EB-special}).
Similarly, it follows from (\ref{eq:B-by-A}), (\ref{eq:EH2-special}), and
(\ref{eq:EA2-special}) that
\begin{align*}
\E[B^2] 
&= 
\E[A^2] + \Pr(A=0)
\\
&=
\lambda (\alpha \E[B]+\tau_0)
+ \lambda^2 (\alpha^2 \E[B^2] + 2 \alpha \tau_0 \E[B] + \tau_0^2)
+ \Pr(A=0).
\end{align*}
We then obtain (\ref{eq:EB2-special}) by rearranging terms of this
equation.
\end{proof}

\begin{lemma}
Under Assumption \ref{assumption:special},
the mean latency $\E[W]$ is given in terms of the probability 
$\pi_0 := \Pr(L=0)$ that the server is idle by
\begin{align}
\E[W] 
&= 
\alpha +\tau_0 
+ 
\frac{\ds \lambda(1+2\lambda\alpha)\left(2\alpha\tau_0+\alpha^2  
+ \frac{(1-\pi_0-\lambda\alpha)\tau_0}{\lambda}
\right)}
{2(1-\lambda^2 \alpha^2)}.
\label{eq:EW-special}
\end{align}
\end{lemma}

\begin{proof}
It follows from (\ref{eq:EW}), (\ref{eq:hatH-special}),
(\ref{eq:EB-special}), and (\ref{eq:EB2-special}) that
\begin{align}
\E[W] 
&= 
\tau_0 + \frac{(1+2\lambda\alpha)\E[B^2] - \E[B]}{2\lambda\E[B]}
\nonumber
\\
&=
\alpha + \tau_0 + \frac{(1+2\lambda\alpha)(\E[B^2] -
\E[B])}{2\lambda\E[B]}.
\label{eq:EW-by-EB-EB2}
\end{align}
Note here that (\ref{eq:EB-special}) and (\ref{eq:EB2-special}) imply
\begin{align}
\frac{\E[B^2] - \E[B]}{\lambda\E[B]} 
&=
\frac{(2\lambda^2\alpha\tau_0+\lambda^2\alpha^2)\E[B]  + \lambda^2 \tau_0^2}
{(1-\lambda^2 \alpha^2)\lambda\E[B]}
\nonumber
\\
&=
\frac{\ds \lambda\left(2\alpha\tau_0+\alpha^2  + \frac{\tau_0^2}{\E[B]}\right)}
{1-\lambda^2 \alpha^2}.
\label{eq:EB2-EB}
\end{align}
In addition, owing to Little's law, the server utilization (i.e., the
mean number of batches being served in steady state) is equal to the
product of the number of batches processed per unit time and the mean
batch processing time:
\begin{equation}
1-\pi_0 = \frac{\lambda}{\E[B]} \cdot \E[H]
=
\lambda\alpha+\frac{\lambda\tau_0}{\E[B]},
\label{eq:pi0-Little}
\end{equation}
where we used (\ref{eq:EH-special}) for the second equality.
Therefore, we obtain (\ref{eq:EW-special}) from 
(\ref{eq:EW-by-EB-EB2}),  (\ref{eq:EB2-EB}), and (\ref{eq:pi0-Little}).
\end{proof}

\begin{remark}
By definition, we have $\pi_0 = \pi(0)$ (see (\ref{eq:pi(z)-def})).
\end{remark}

Even under Assumption \ref{assumption:special}, it seems difficult to
determine the exact value of $\pi_0$.
However, we have the following simple lower bound for this quantity:

\begin{lemma}
\label{lemma:pi0-bound}

Under Assumption \ref{assumption:special},
$\pi_0$ is bounded below by
\begin{equation}
\pi_0 \geq \max(0, 1-\lambda(\alpha+\tau_0)).
\label{eq:pi0-bound}
\end{equation}
\end{lemma}
\begin{proof}
Because $B \geq 1$ holds with probability one, $\E[B] \geq 1$ holds.
We then have from (\ref{eq:pi0-Little}),
\[
\pi_0 \geq 1-\lambda (\alpha+\tau_0),
\]
which and $\pi_0 \geq 0$ imply (\ref{eq:pi0-bound}).
\end{proof}
\begin{remark}
If $\lambda(\alpha+\tau_0) < 1$, the quantity 
$1-\lambda(\alpha+\tau_0)$ is equal to the probability that the server
is idle in a stationary single-service M/D/1 queue with the arrival rate
$\lambda$ and the processing time $H^{[1]} = \alpha+\tau_0$, where 
arriving inference jobs are processed one by one.
\end{remark}

\begin{remark}
It follows from (\ref{eq:pi0-Little}) that
$\E[B] \geq \max(1,\lambda\tau_0/(1-\lambda\alpha))$,
so that if Assumption \ref{assumption:c^b} is satisfied, we have
(cf.\ (\ref{eq:eta}))
\begin{equation}
\eta 
\geq 
\frac{1}{\beta + c_0/\max(1,\lambda\tau_0/(1-\lambda\alpha))}.
\label{eq:eta-bound}
\end{equation}
\end{remark}

We are in a position to obtain the main result of this paper:
\begin{theorem}
\label{theorem:EW-special}

Under Assumption \ref{assumption:special},
the mean latency $\E[W]$ is bounded above by
\begin{align}
\E[W] 
&\leq
\frac{\alpha+\tau_0}{2(1-\lambda \alpha)}
\left(
1 + 2\lambda\tau_0 + \frac{1 - \lambda\tau_0}{1+\lambda\alpha}
\right)
=:
\phi_0(\lambda,\alpha,\tau_0),
\label{eq:EW-bound-0}
\\
\mbox{and}\hspace{1em}
\nonumber
\\
\E[W] 
&\leq
\frac{3}{2}
\cdot
\frac{\ds \tau_0}
{1-\lambda \alpha}
+
\frac{\alpha}{2}
\cdot
\frac{\ds \lambda\alpha +2}
{1-\lambda^2 \alpha^2}
=:
\phi_1(\lambda,\alpha,\tau_0).
\label{eq:EW-bound-1}
\end{align}
In addition, we have $\phi_0(\lambda,\alpha,\tau_0) \leq \phi_1(\lambda,\alpha,\tau_1)$
if and only if $\lambda \leq 1/(\alpha+\tau_0)$.
\end{theorem}
\begin{proof}
As stated in Lemma \ref{lemma:pi0-bound}, we have two lower bounds 
for $\pi_0$. Using $\pi_0 \geq 1-\lambda(\alpha+\tau_0)$ and
(\ref{eq:EW-special}), we have
\begin{align*}
\E[W] 
&\leq
\alpha+ \tau_0 + \frac{\ds \lambda(1+2\lambda\alpha)(\alpha+\tau_0)^2}
{2(1-\lambda^2 \alpha^2)}
\\
&=
\frac{\alpha+\tau_0}{2(1-\lambda \alpha)}
\cdot
\frac{2 + \lambda\alpha + \lambda\tau_0 + 2\lambda^2\alpha\tau_0}
{1+\lambda\alpha},
\end{align*}
which implies (\ref{eq:EW-bound-0}). On the other hand,
we have (\ref{eq:EW-bound-1}) from $\pi_0 \geq 0$ and (\ref{eq:EW-special}) 
as follows:
\begin{align*}
\E[W] 
&\leq
\alpha +\tau_0
+ 
\frac{\ds (1+2\lambda\alpha)
\left(\lambda\alpha\tau_0 +\lambda\alpha^2 + \tau_0\right)}
{2(1-\lambda^2 \alpha^2)}
\\
&=
\alpha + \tau_0
+ 
\frac{\ds (1+2\lambda\alpha)\tau_0}
{2(1-\lambda \alpha)}
+
\frac{\ds (1+2\lambda\alpha)\lambda\alpha^2}
{2(1-\lambda^2 \alpha^2)}
\\
&=
\frac{\ds 3\tau_0}
{2(1-\lambda \alpha)}
+
\frac{\ds \lambda\alpha^2 +2\alpha}
{2(1-\lambda^2 \alpha^2)}.
\end{align*}
The relation $\phi_0(\lambda,\alpha,\tau_0) \leq \phi_1(\lambda,\alpha,\tau_1)
\Leftrightarrow \lambda \leq 1/(\alpha+\tau_0)$ is thus obvious from these
derivations.
\end{proof}

Theorem \ref{theorem:EW-special} provides a surprisingly simple upper
bound for the mean latency $\E[W]$. For convenience, let
\begin{equation}
\phi(\lambda,\alpha,\tau_0) 
:= 
\min(\phi_0(\lambda,\alpha,\tau_0), \phi_1(\lambda,\alpha,\tau_0)).
\label{eq:phi}
\end{equation}
Even though this upper bound is obtained by replacing the idle
probability $\pi_0$ with its almost trivial lower bound in (\ref{eq:pi0-bound}),
it provides a quite good approximation to the exact value of the mean
latency $\E[W]$, as we will see in the next section.

\section{Numerical Evaluation}
\label{sec:experiments}

In this section, we numerically validate the applicability
of the derived mathematical formulas to GPU-based inference servers.
Throughout this section, we concentrate on the model considered in
Section \ref{ssec:special_case} where processing times are 
deterministic and linearly increase with the batch size.
We first conduct simulation experiments to examine how well the exact
mean latency $\E[W]$ is approximated by the closed-form upper bound
$\phi(\lambda,\alpha,\tau_0)$, assuming both infinite and
finite maximum batch sizes. For the simulation experiments, we 
employ the model parameters $\alpha$ and $\tau_0$
estimated in Section \ref{ssec:special_case} from Table
\ref{table:resnet}; see the paragraph just after Assumption
\ref{assumption:special} for more detail.
We then show the usefulness of the closed-form expression
of $\phi(\lambda,\alpha,\tau_0)$ by comparing it to the measured
latency in real-world implementation of GPU-based inference servers.
We use two types of GPUs for this experiment, NVIDIA Tesla V100
and Tesla T4, which are available on the Amazon Elastic Compute Cloud
(Amazon EC2). Also, the MLPerf inference benchmark
\cite{mlperf} is used for load generation and time measurement.
The GPU inference servers are set up using NVIDIA's implementation of
the MLPerf inference benchmark v0.5, which is publicly available on GitHub
\cite{github}.

Fig.\ \ref{fig:infinite_EW} shows simulation results for
the mean latency $\E[W]$ and its upper bounds
$\phi_0(\lambda,\alpha,\tau_0)$ and $\phi_1(\lambda,\alpha,\tau_0)$
given in (\ref{eq:EW-bound-0}) and (\ref{eq:EW-bound-1})
(recall that the normalized load $\rho$ is defined in
(\ref{eq:rho})). We observe that the combination (\ref{eq:phi}) of
these upper bounds  quite well approximates the exact curve of $\E[W]$.
In particular, except for small values of $\rho$, 
$\E[W]$ takes fairly close values to $\phi_1(\lambda,\alpha,\tau_0)$.
The reasons for this remarkable accuracy of approximation can be
explained as follows.
Recall that the upper bound $\phi_1(\lambda,\alpha,\tau_0)$ is
obtained by replacing the idle probability $\pi_0$ with its trivial
lower bound $0$. In Fig.\ \ref{fig:infinite_pi0}, the server
utilization $1-\pi_0$ is plotted as a function of the normalized load
$\rho$. As a reference, we also plot its upper bound 
$\min(1,\lambda(\alpha+\tau_0))$ (cf.\ (\ref{eq:pi0-bound})).
From this figure, we see that \textit{the server utilization takes
a value close to $1$ even for a moderate value of $\rho$},
which is quite different from ordinary single-server queues, 
where the server utilization is equal to the traffic intensity.
This phenomenon comes from the fact that the server's
processing speed largely increases as the batch size increases,
so that \textit{the system is overloaded for small batch sizes} 
even under a moderate load level $\rho$. Because of this behavior of the
server utilization, the upper bound $\phi_1(\lambda,\alpha,\tau_0)$ is
a good approximation to the mean latency $\E[W]$ for a wide 
range of $\rho$. 
On the other hand, for small $\rho$, the upper bound 
$\phi_0(\lambda,\alpha,\tau_0)$ is a good approximation to
$\E[W]$. Note that $\phi_0(\lambda,\alpha,\tau_0)$ is obtained by
replacing the mean batch size $\E[B]$ with its trivial lower bound
$1$. Therefore, $\E[W] \simeq \phi_0(\lambda,\alpha,\tau_0)$ implies
that the mean batch size $\E[B] \simeq 1$, i.e., the server does not
sufficiently leverage its batch-processing capability in such a region.

\begin{figure}[tp]
\centering
\begin{subfigure}[t]{0.48\textwidth}
\centering
\includegraphics[scale=1.0]{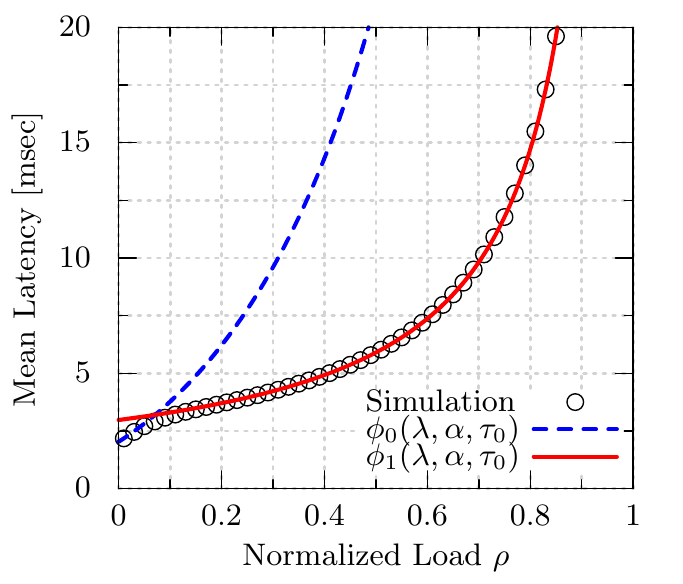}
\caption{Tesla V100, Mixed precision (Table \ref{table:resnet} (a)).}
\end{subfigure}
\begin{subfigure}[t]{0.48\textwidth}
\centering
\includegraphics[scale=1.0]{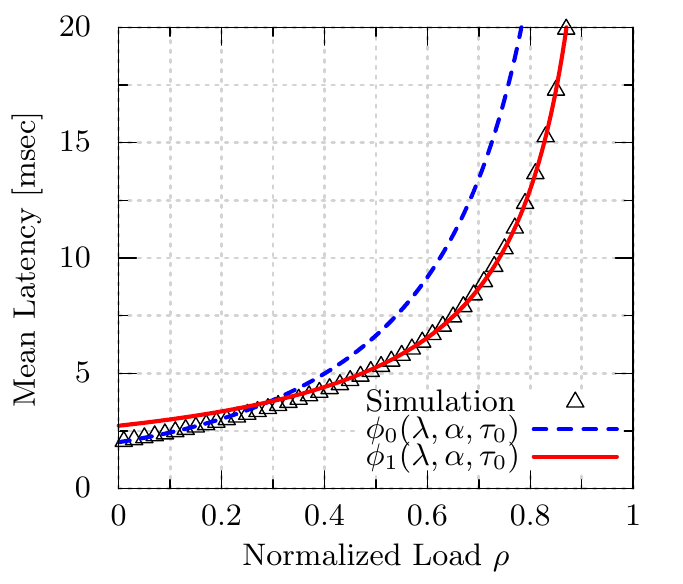}
\caption{Tesla P4, INT8 (Table \ref{table:resnet} (b)).}
\end{subfigure}
\caption{The mean latency $\E[W]$ and its upper bound (in milliseconds).}
\label{fig:infinite_EW}
\end{figure}

\begin{figure}[tp]
\centering
\begin{subfigure}[t]{0.48\textwidth}
\centering
\includegraphics[scale=1.0]{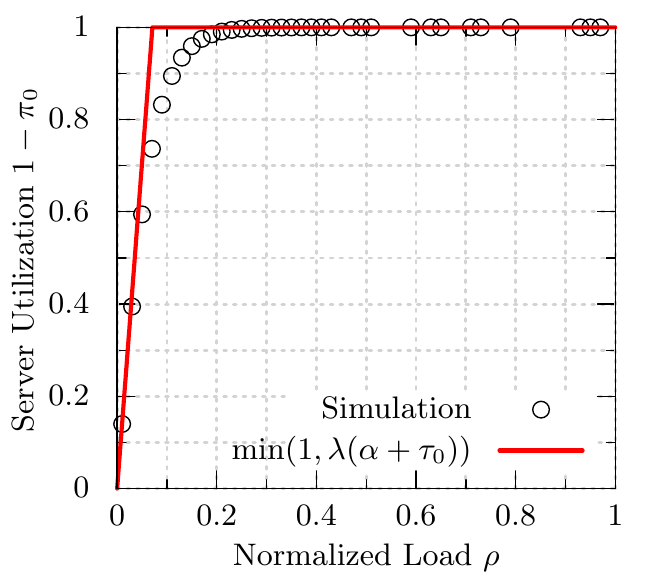}
\caption{Tesla V100, Mixed precision (Table \ref{table:resnet} (a)).}
\end{subfigure}
\begin{subfigure}[t]{0.48\textwidth}
\centering
\includegraphics[scale=1.0]{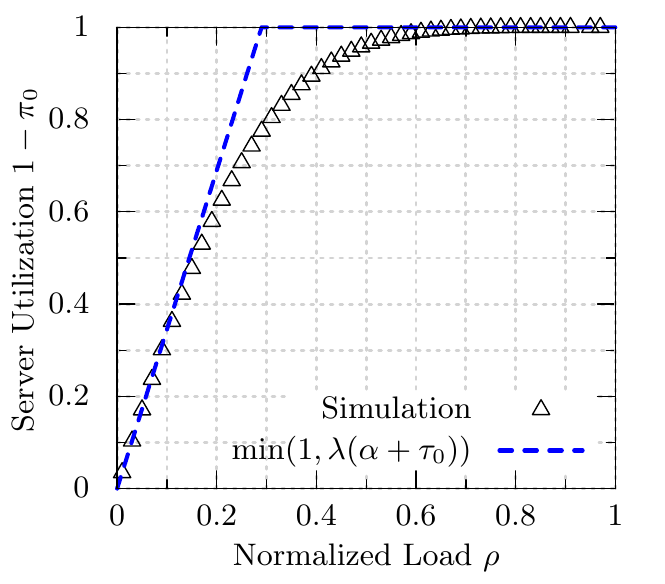}
\caption{Tesla P4, INT8 (Table \ref{table:resnet} (b)).}
\end{subfigure}
\caption{The server utilization $1-\pi_0$ and its upper bound.}
\label{fig:infinite_pi0}
\mbox{}
\\
%
\centering
\includegraphics[scale=1.0]{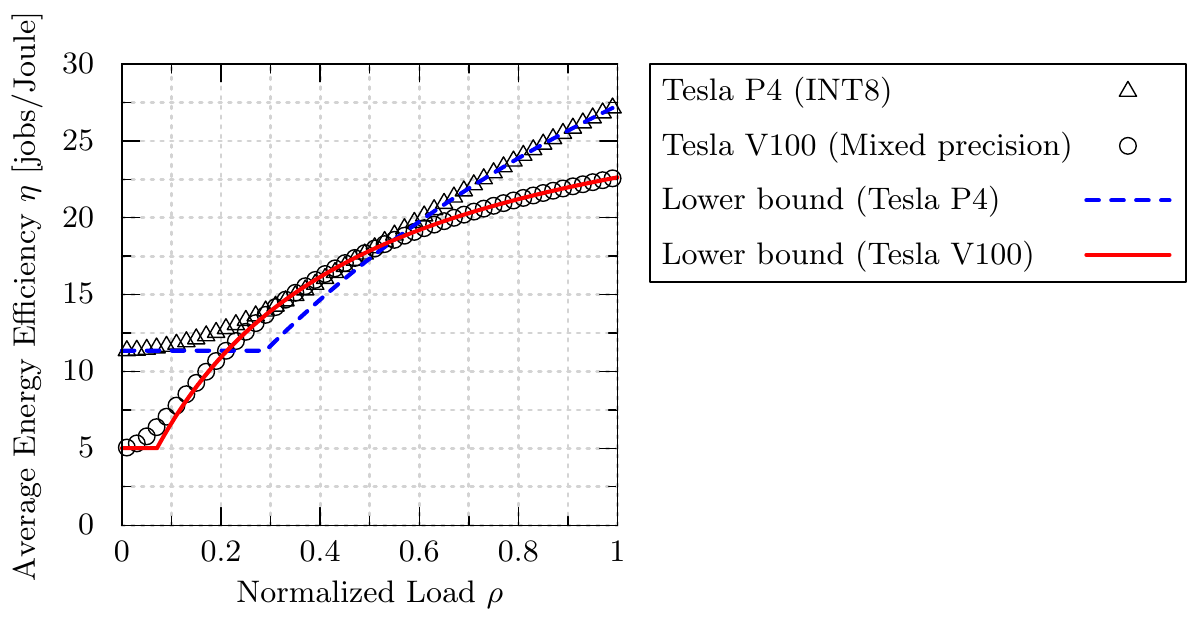}
\caption{The average energy efficiency $\eta$ and its lower bound.}
\label{fig:eta}
\mbox{}\\
\vspace{1ex}
\includegraphics[scale=1.0]{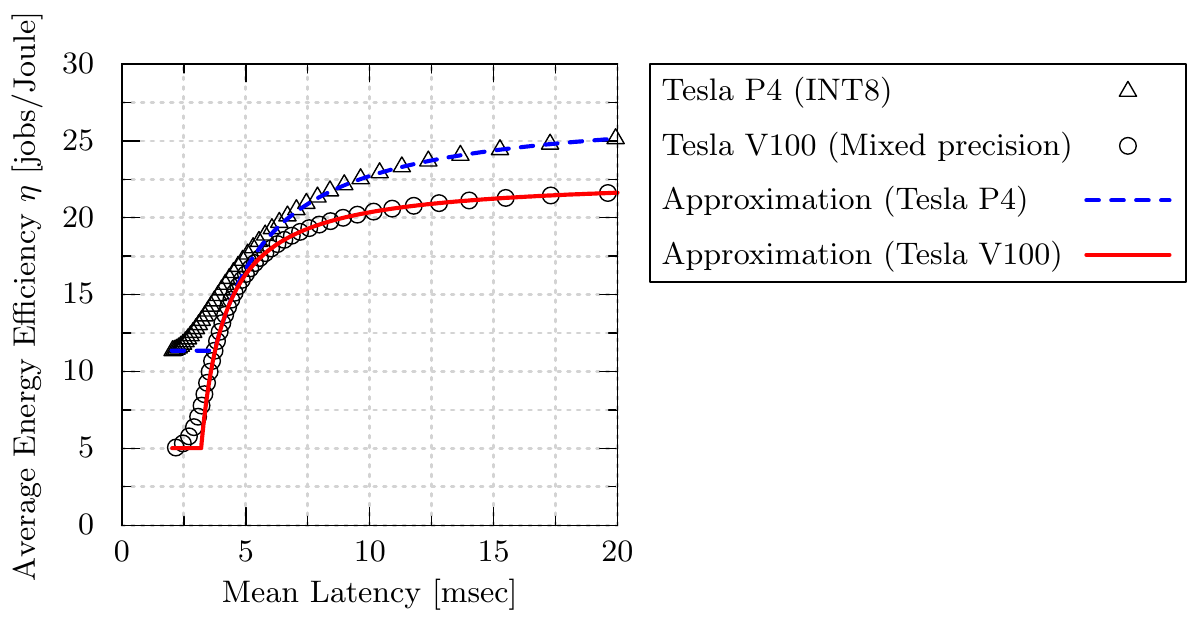}
\caption{The parametric curve of the average energy efficiency $\eta$
and the mean latency $\E[W]$ with parameter $\rho$. The
approximation curves are plotted using (\ref{eq:eta-bound}) and
(\ref{eq:phi}).
}
\label{fig:infinite_tradeoff}
\end{figure}

We next discuss the energy efficiency using the linear model
(\ref{eq:c^b}) considered in Section \ref{ssec:energy}.
Recall that the average energy efficiency $\eta$ is defined as
(\ref{eq:eta-def}), which represents the mean number of jobs processed
with unit energy. In Fig.\ \ref{fig:eta}, simulation results for
$\eta$ and its lower bound (\ref{eq:eta-bound}) are plotted as functions of
the normalized load $\rho$. From this figure, we observe that the
energy efficiency can be largely enhanced by letting the server
adequately loaded. Also, the energy-efficiency $\eta$ is
well-approximated by the lower bound (\ref{eq:eta-bound}) except for
small values of $\rho$. Fig.\ \ref{fig:infinite_tradeoff} shows the
energy-latency tradeoff, where the relation between $\eta$ and the
mean latency $\E[W]$ is plotted with parameter $\rho$.
In this figure, we also plot approximation curves obtained by
combining (\ref{eq:eta-bound}) and (\ref{eq:phi}).
We see that the closed-form bounds (\ref{eq:eta-bound}) and
(\ref{eq:phi}) are useful to determine an adequate operating point
of the server, taking the energy-latency tradeoff into consideration.

We then discuss the relation between the model considered in
this paper and a corresponding batch-service queue with 
\textit{finite maximum batch size} $b_{\max}$. As mentioned in Section
\ref{sec:intro}, the mean latency in the case of finite $b_{\max}$ 
can be numerically obtained with results in \cite[Section
4.2]{Neuts89}. Fig.\ \ref{fig:finite} shows that if $b_{\max}$ is
sufficiently large, the mean latency is well approximated by our
closed-form upper bound $\phi(\lambda,\alpha,\tau_0)$ given by
(\ref{eq:phi}). If $b_{\max}$ is small, on the other hand, the mean
latency deviates from $\phi(\lambda,\alpha,\tau_0)$ for the arrival
rate $\lambda$ near the stability boundary 
$\lambda = \mu^{[b_{\max}]} = b_{\max}/(\alpha b_{\max}+\tau_0)$.
However, we observe from this figure that even for small values of
$b_{\max}$, the mean latency is still well-approximated by
(\ref{eq:phi}) if the system is moderately loaded, i.e., $\lambda$ is
sufficiently small compared to $\mu^{[b_{\max}]}$.

\begin{figure}[tp]
\centering
\begin{subfigure}{0.98\textwidth}
\centering
\includegraphics[scale=1.05]{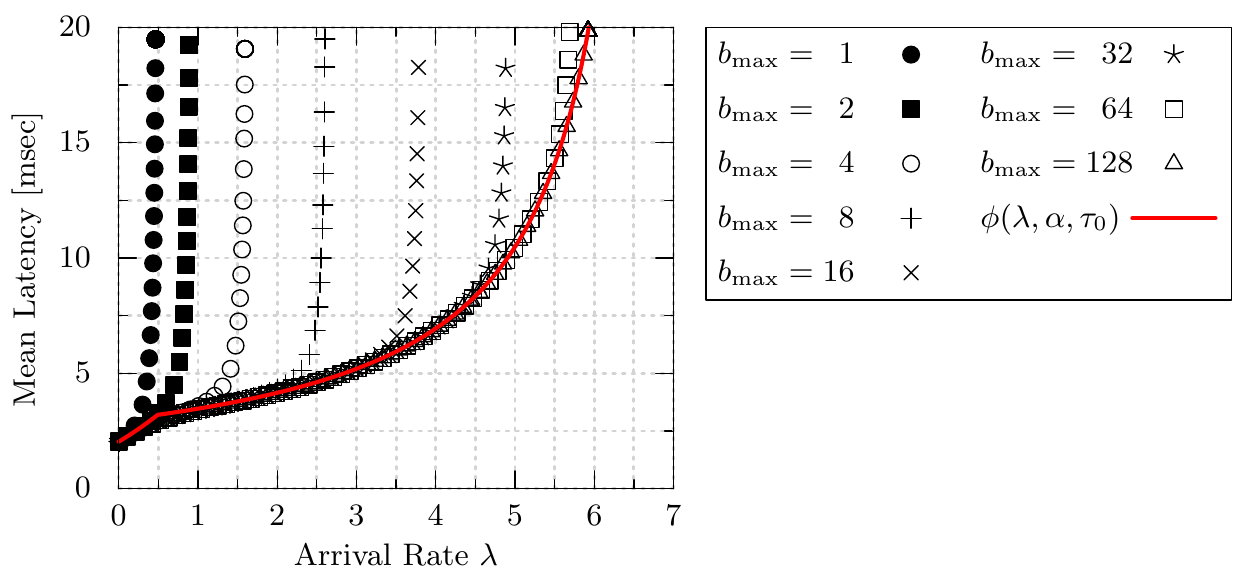}
\caption{Tesla V100, Mixed precision (Table \ref{table:resnet} (a)).}
\end{subfigure}
\mbox{}\vspace{1ex}\\
\begin{subfigure}{0.98\textwidth}
\centering
\includegraphics[scale=1.05]{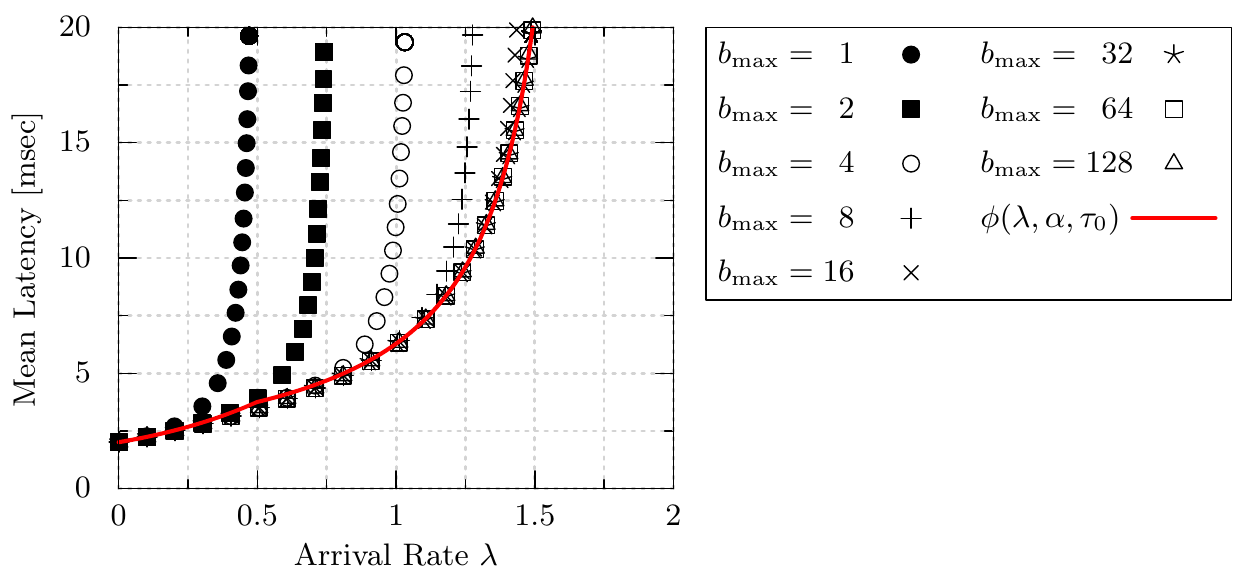}
\caption{Tesla P4, INT8 (Table \ref{table:resnet} (b)).}
\end{subfigure}
\caption{The mean latency in the case of finite maximum batch sizes
$b_{\max}$, plotted with the closed-form upper bound
$\phi(\lambda,\alpha,\tau_0)$ for the case of $b_{\max}=\infty$.}
\label{fig:finite}
\end{figure}

Finally, we compare the derived closed-form expression
$\phi(\lambda,\alpha,\tau_0)$ to the mean latency measured in a real
GPU-based inference server implemented with MLPerf inference
benchmark. Three types of networks are used in this experiment:
MobileNet, ResNet50, and SSD-MobileNet. 
We use the two GPUs, Tesla V100 and Tesla T4 as mentioned
above. We set the computing precision in Tesla V100 to FP16 and that
in Tesla T4 to INT8. We also set the maximum batch size $b_{\max}=64$
in all experiments. First, the batch processing time for each batch
size $b$ is measured using MultiStream Scenario of MLPerf inference
benchmark. For each $b$, we collect $100$ samples of batch processing
times and we use their median as the representative value. Fig.\
\ref{fig:mlperf-H} shows the batch processing time as a function of
the batch size for the six cases in total, where $\alpha$ and
$\tau_0$ are fitted by the least squares method. We observe that batch
processing times linearly increase with the batch size as we 
assumed in the mathematical analysis. In ResNet50, however, 
we observe additional stair-like increases in the processing time at
several places. Fig.\ \ref{fig:mlperf-mu} shows the corresponding
throughput curves plotted as functions of the batch size. 
We observe that the stair-like increases in processing times
result in discontinuous decreases in the throughput.

\begin{figure}[tp]
\centering
\begin{subfigure}{0.98\textwidth}
\centering
\includegraphics[scale=1.0]{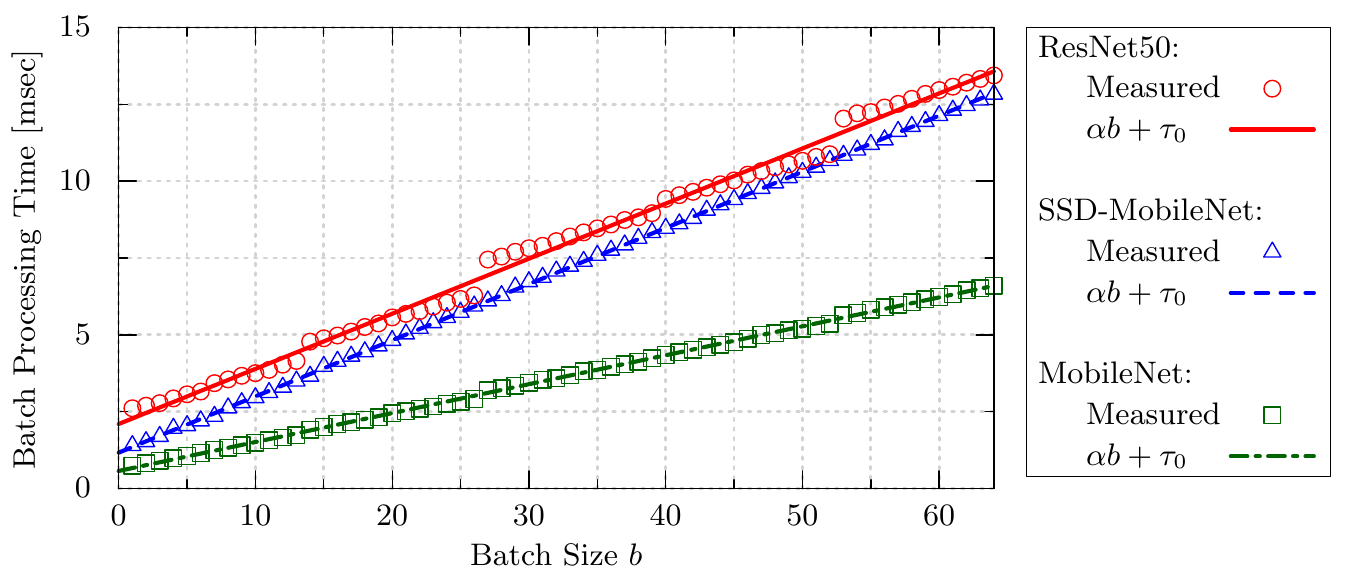}
\caption{Tesla V100, FP16.}
\end{subfigure}
\mbox{}\vspace{1ex}\\
\begin{subfigure}{0.98\textwidth}
\centering
\includegraphics[scale=1.0]{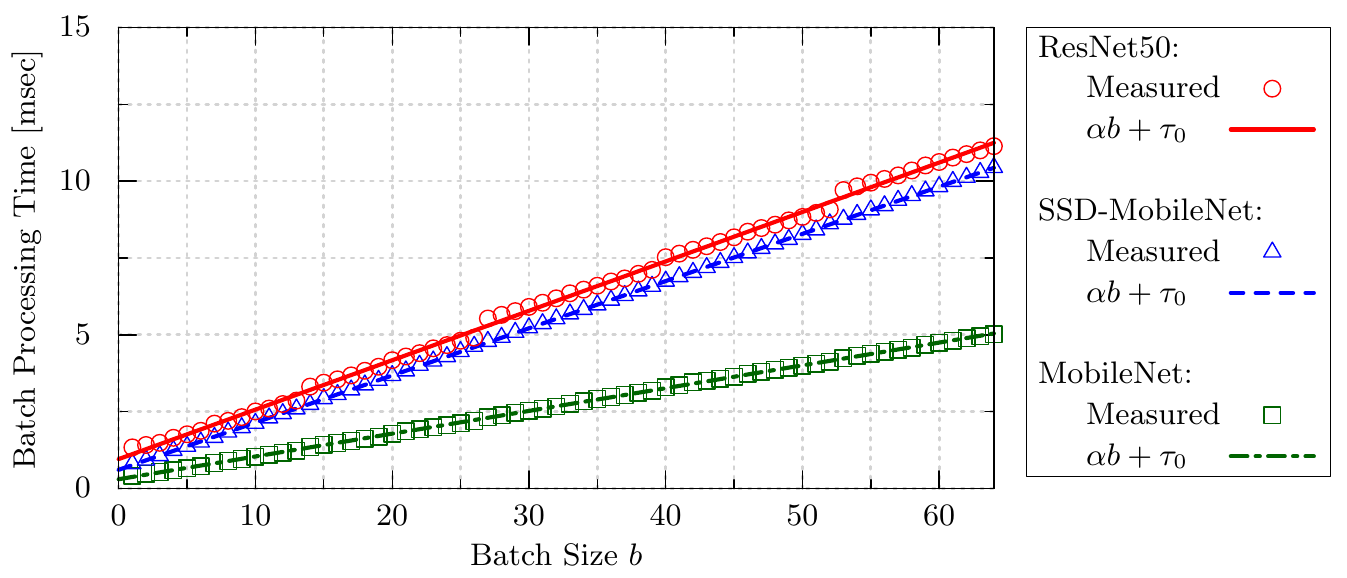}
\caption{Tesla T4, INT8.}
\end{subfigure}
\caption{The median of batch processing times measured using MLPerf
MultiStream Scenario.}
\label{fig:mlperf-H}
\end{figure}

\begin{figure}[tp]
\centering
\begin{subfigure}{0.98\textwidth}
\centering
\includegraphics[scale=1.0]{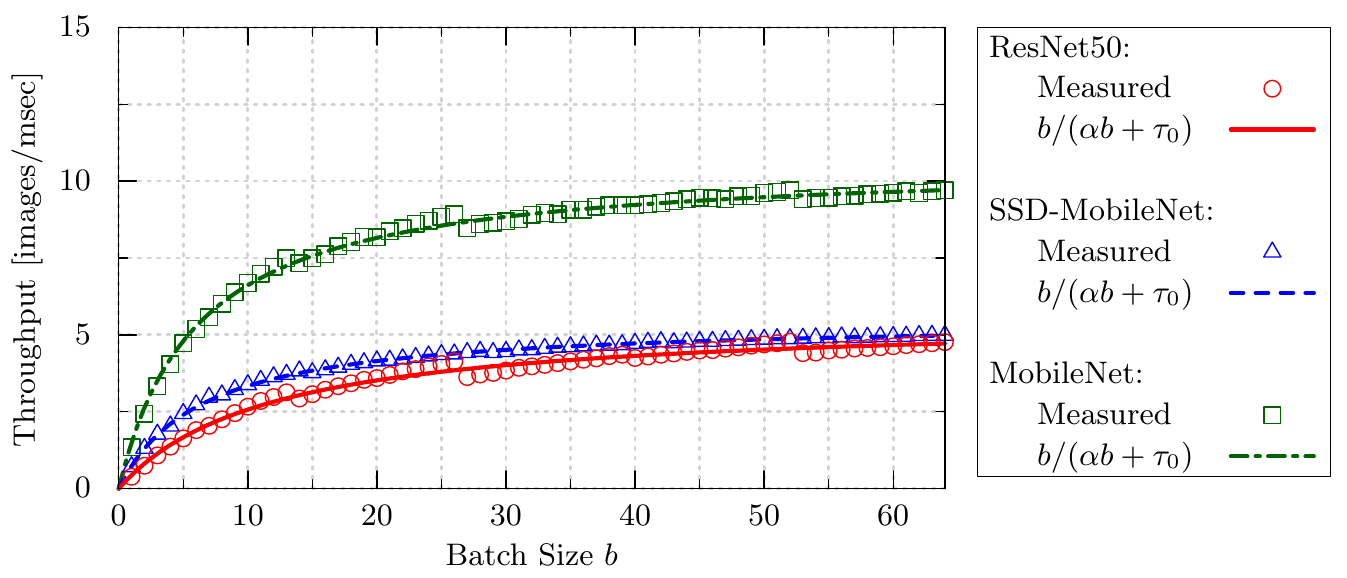}
\caption{Tesla V100, FP16.}
\end{subfigure}
\mbox{}\vspace{1ex}\\
\begin{subfigure}{0.98\textwidth}
\centering
\includegraphics[scale=1.0]{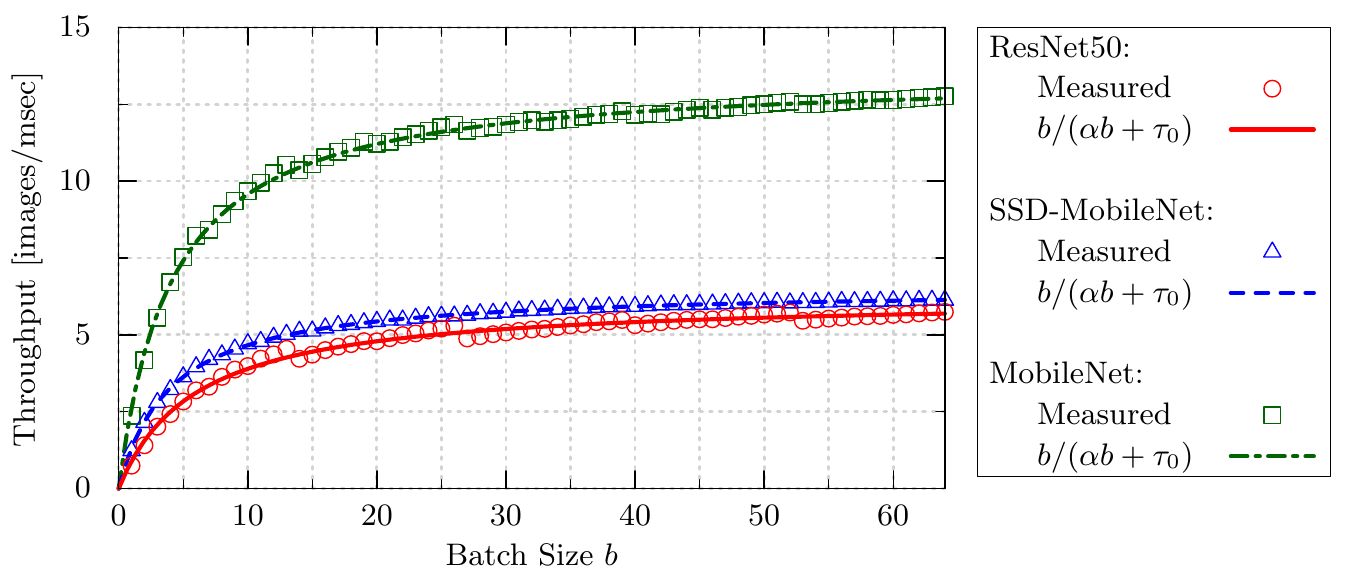}
\caption{Tesla T4, INT8.}
\end{subfigure}
\caption{The median of batch processing times measured using MLPerf
MultiStream Scenario.}
\label{fig:mlperf-mu}
\end{figure}

\begin{figure}[tp]
\centering
\begin{subfigure}{0.98\textwidth}
\centering
\includegraphics[scale=1.0]{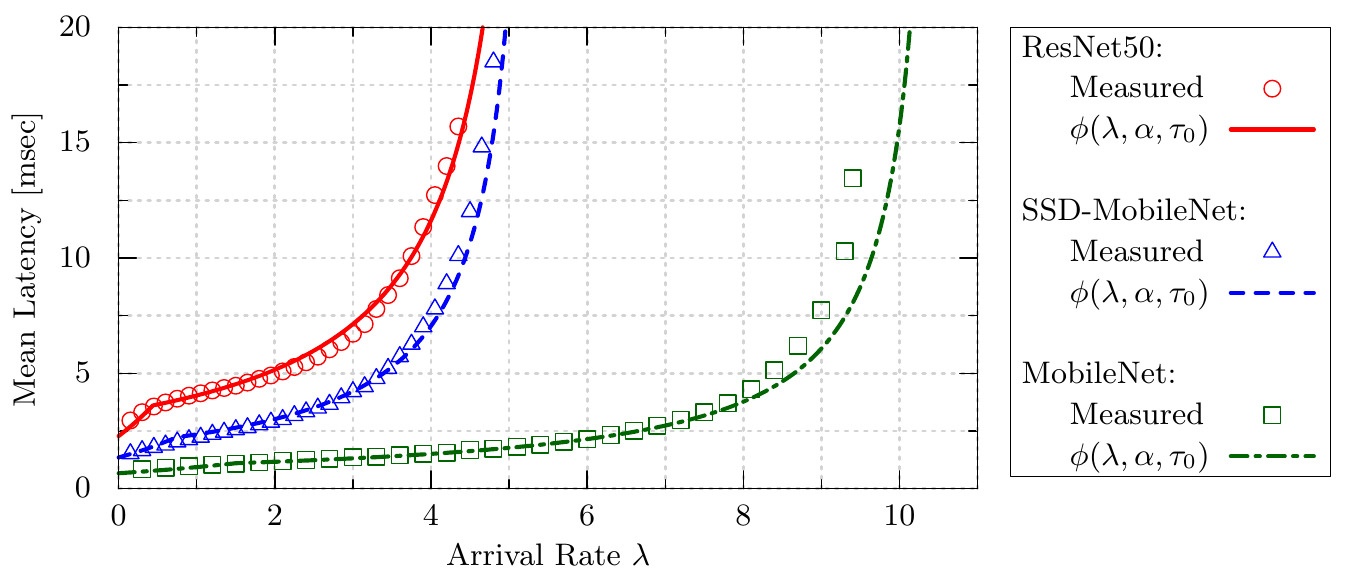}
\caption{Tesla V100, FP16.}
\end{subfigure}
\mbox{}\vspace{1ex}\\
\begin{subfigure}{0.98\textwidth}
\centering
\includegraphics[scale=1.0]{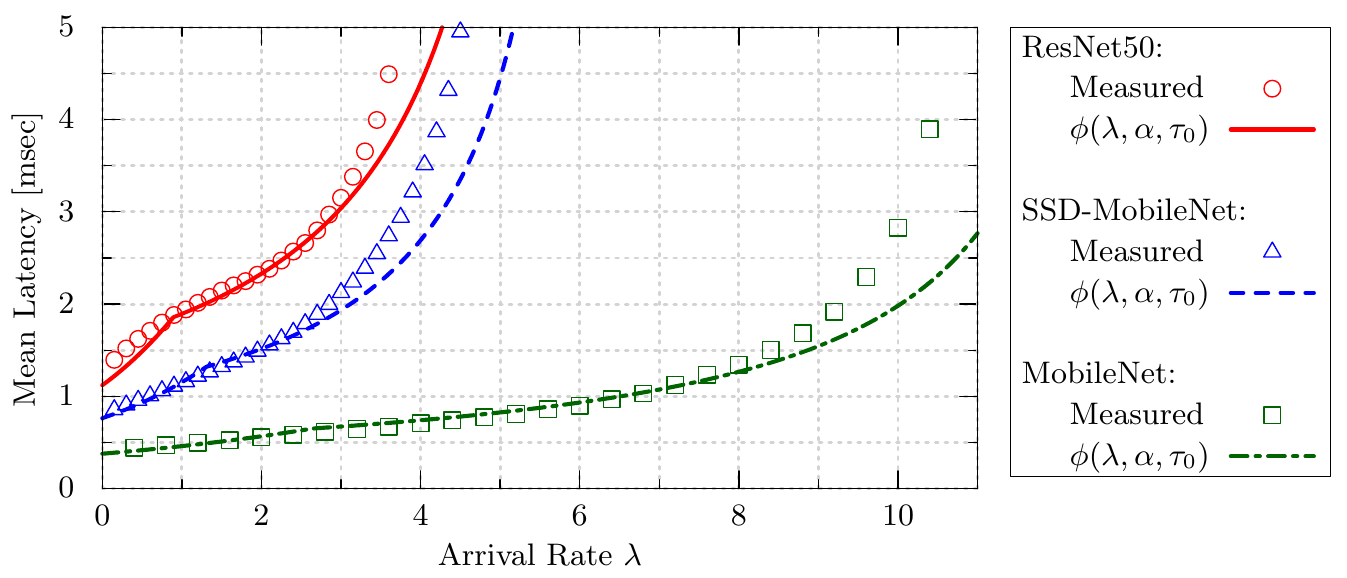}
\caption{Tesla T4, INT8.}
\end{subfigure}
\caption{The mean latency measured using MLPerf Server Scenario,
plotted with the closed-form upper bound $\phi(\lambda,\alpha,\tau_0)$.}
\label{fig:mlperf-latency}
\end{figure}

Regardless of such an irregular behavior of batch processing times,
the latency performance is still well explained by the mathematical
formula we have derived in this paper. We use MLPerf Server Scenario
to measure the mean latency $\E[W]$, where the load generator sends
inference requests to the GPU server according to a Poisson process of
a given arrival rate $\lambda$. The measurement for each parameter are
run for 10 seconds. Fig.\ \ref{fig:mlperf-latency} shows
the mean latency plotted as a function of the arrival rate of
requests. We observe that each latency curve is well explained by the 
closed-form expression $\phi(\lambda,\alpha,\tau_0)$, except for the
region near the stability boundary. Also, we observe that in Tesla T4,
the measured mean latency starts to deviate from the theoretical curve
$\phi(\lambda,\alpha,\tau_0)$ from a relatively small value of
$\lambda$. This phenomenon is due to a hardware limit of the Tesla T4:
this GPU is designed to operate at a low power of 70 [W], which in
turn causes the operating clock to be forcibly lowered (by the SW
Power Cap mechanism) when the computational load becomes excessive. 
Since it is not reasonable to operate the server under such an
excessive load, we can conclude that for the both GPUs and all networks
considered, our closed-form formula explains the mean latency quite
well in practical operating ranges.

\section{Conclusion}
\label{sec:conclusion}

In this paper, we introduced a queueing model representing GPU-based
inference servers with dynamic batching. We modeled an inference
server as a batch-service queueing model with infinite maximum batch
sizes and batch-size dependent processing times. 
We first showed that the energy efficiency of the server increases
with the arrival rate of inference jobs, which suggests that it is
energy-efficient to operate the server under a traffic load as large
as possible, within a latency requirement of inference jobs.
We then derived a simple closed-form upper bound for the mean latency in
Theorem \ref{theorem:EW-special}, under the assumption that 
the batch processing time linearly increases with the batch size.
Through numerical and simulation experiments, we showed that the exact
value of the mean latency is well-approximated by this simple upper bound.
We further compared this formula with the latency curve measured in 
real implementation of GPU-based inference servers, which showed that
the real performance curve is also well explained by the derived simple formula.

\section*{Acknowledgements}
The author would like to thank the anonymous reviewers for their
helpful comments. This work was supported in part by JSPS KAKENHI
Grant Number 18K18007.





\end{document}